\documentclass[sn-mathphys]{sn-jnl}
\usepackage{graphicx}%
\usepackage{multirow}%
\usepackage{amsmath,amssymb,amsfonts}%
\usepackage{amsthm}%
\usepackage{mathrsfs}%
\usepackage[title]{appendix}%
\usepackage{xcolor}%
\usepackage{textcomp}%
\usepackage{manyfoot}%
\usepackage{booktabs}%
\usepackage{algorithm}%
\usepackage{algorithmicx}%
\usepackage{algpseudocode}%
\usepackage{listings}%
\usepackage{subcaption}
\usepackage{svg}     % loads graphicx automatically
\svgpath{images/}
\usepackage{enumitem}
\usepackage{comment}
\usepackage{multicol}
\usepackage{tikz}
\usepackage{amsmath}
\usepackage{amssymb}
\usepackage{amsthm}
\usepackage{algorithm}
\usepackage{color}
\usepackage{bbm}
\usepackage{graphicx}
\newtheorem{thm} {Theorem}
\newtheorem{lemma} {Lemma}
\newtheorem{scollium} {Scollium}

\newtheorem{prop}[thm]{{\bf Proposition}}

\newtheorem{cor}[thm]{{\bf Corollary}}
\newcommand{\Set}{Set}
\newcommand{\Surf}{Surface}
\newcommand{\R}{\mathbb{R}}
\newcommand{\ECT}{ECT}

\newcommand{\Z}{\mathbb{Z}}
\newtheorem{definition}{\bf Definition}

\theoremstyle{remark}
\newtheorem*{remark}{Remark}

\title{The Reeb Transform}

\author[1]{\fnm{Erin} \sur{Chambers}}\email{echambe2@nd.edu}

\author[1]{\fnm{Shankha Shubhra} \sur{Mukherjee}}\email{smukher4@nd.edu}
\author[2]{\fnm{Katharine} \sur{Turner}}\email{katharine.turner@anu.edu.au}
\affil[1]{\orgdiv{Computer Science and Engineering}, \orgname{University of Notre Dame}, \orgaddress{\city{Notre Dame}, \postcode{46556}, \state{Indiana}, \country{USA}}}

\affil[2]{\orgdiv{Mathematical Sciences Institute}, \orgname{Australian National University}, \orgaddress{\city{Canberra}, \postcode{2601}, \state{ACT}, \country{Australia}}}

\begin{document}

\title[The Reeb transform]{The Reeb Transform}
\abstract{
%How much information does the Reeb transform encode for a shape?  For some classes of shapes it is invertible, sometimes it is not.  In this paper, we prove both positive results as well as providing counterexamples for the injectivity of the Reeb transform for other classes of shapes.

Reeb Transforms offer a compact representation of how shapes in $\mathbb{R}^n$ change when sliced along varying directions. We define the Reeb Transform as the family of Reeb graphs induced by height functions along all directions in the unit sphere. In this paper, we develop a rigorous treatment of Reeb Transforms for o-minimal definable sets, with particular emphasis on one-dimensional stratified spaces embedded in $\mathbb{R}^d$ and on surfaces in $\mathbb{R}^3$. We establish injectivity of the Reeb Transform in multiple settings, including compact surfaces in $\mathbb{R}^3$, capturing the essential topological features needed to uniquely reconstruct such surfaces. 
%We further show that local modifications, under certain regularity conditions, do not alter the Reeb Transform, thereby demonstrating a form of local stability. 
However, in dimensions above three, the Reeb Transform ceases to be injective, indicating the limitations of this descriptor in higher-dimensional settings. 
%In addition to these main results, we develop several other properties of Reeb Transforms that underscore their versatility and significance for shape analysis within o-minimal frameworks.
}



%NEED TO UPDATE

\maketitle

%\newpage
\section{Introduction}
The fundamental goal in statistical shape analysis is to define and compute meaningful distances between diﬀerent subsets of Euclidean space. A recent landmark-free approach to quantify both the geometry and topology of a shape is to use a topological transform such as the Persistent Homology Transform (PHT) or the Euler Characteristic Transform (ECT). Both of these transforms
take a shape M, viewed as a subset $\R^d$, and associate to each direction $v\in S^{d-1}$ a shape summary obtained by scanning $M$ in the direction $v$, calculating the
persistent homology $(PH(M,v))$ and the Euler curve respectively.
Diﬀerent formulations of the PHT and ECT have been demonstrably useful in diverse applications including prediction of disease progression from the
shapes of tumours \cite{crawford2020predicting} to the quantification of morphological variation of barley seeds~\cite{amezquita2022measuring}.
A nice introductory article is \cite{munch2025invitation}.

This paper introduces a variant of this topological transform, called the Reeb Transform, where we replace Euler curves and persistence modules with Reeb graphs.
Reeb graphs were originally introduced by Georges Reeb in the context of Morse theory~\cite{reeb1946}, providing a powerful way to capture the evolution of level sets of a real-valued functions on topological spaces~\cite{a355ff70-21ad-3b86-9c5b-ad05714c2715, Borodzik_2016, biasotti2000extended, dey2011reeb}.
%[Why the manifold restriction!! - https://hal.science/hal-01633101v2/document ]
Reeb graphs encode structural information about how sets decompose along slicing planes, but only at the level of connected components.  In contrast, other topological signatures like Euler transforms or persistence diagrams can capture more information about higher dimensional features, but lack information about connectivity that the Reeb graph can capture.
If $n_k$ denotes the number of $k$-cells in a cell decomposition of $M$, then the computational cost of extracting a Reeb graph is
$O\bigl((n_0+n_1)\log(n_0+n_1)\bigr)$, which is the same asymptotic complexity as computing degree-zero persistent homology. %We can compare this to $O((n_k+n_{k+1})^3)$ the computation of degree $k$ persistent homology and $O(n_0+n_1+n_2+\ldots n_{\text{maximum dimension}})$. 
While it may seem counterintuitive, the linear computation cost of Euler characteristic can be substantially more than the cost of computing the Reeb graph, which only uses the $0$ and $1$ cells. We can compare this with the computational cost $O\bigl((n_k+n_{k+1})^3\bigr)$ for computing degree-$k$ persistent homology, and with
$O\bigl(n_0+n_1+\cdots+n_m\bigr)$ for computing the Euler characteristic, where $m$ is the maximum dimension of a cell in the decomposition of $M$.
 %An added constraint is dimensionality, as persistence captures information about higher homology features which are lost in Reeb graphs, as they only capture homology information from lower dimensional features.
Over the past few decades Reeb graphs have become a cornerstone in topological data analysis and computational geometry, finding applications in shape matching, data visualization, and manifold learning~\cite{Biasotti2008, 10.1111:cgf.12933, DBLP:journals/corr/abs-2106-00157, NIPS2011_3a077244}. Older surveys~\cite{Biasotti2008} and more recent overviews~\cite{Yan2021, Bollen2021} document how Reeb graphs have been deployed in both theoretical settings and practical visualization tasks, underscoring their versatility.

%The answer becomes more interesting than in the case of the ECT or PHT - as we change prove that when restricted to some families of shapes, the Reeb transform is a sufficient statistics, but that in general we can prove that certain modifications of Euclidean subsets leave the  the Reeb Transform invariant.

%Much like other topological descriptions (e.g., persistent homology or Euler integral transforms~\cite{Turner2014, Curry2018, ghrist2018persistenthomologyeulerintegral}), 

It is well known that both the ECT and the PHT are injective for o-minimal sets, meaning the shape is uniquely determined by these descriptors. This further motivates their utility as a shape descriptor with provable quality~\cite{Turner2014,Curry2018}, although an exponential number of directions may be necessary to uniquely determine the shape~\cite{DBLP:journals/corr/abs-1912-12759}.
It is thus natural to ask the following question: is the Reeb transform invertible?  In this paper, we will see that unfortunately the Reeb transform is not invertible for the widest class of subsets of Euclidean shapes. The issue is that we do not encode the information about the homology of level sets, just their connected components.  

 %Consider the sphere vs the ball vs a spherical shell. The Reeb graphs are identical. When given the connected components of a level set we cannot distinguish between a level set being a disc, an annulus, or a circle. See Figure~\ref{fig:coneAndSphere}.

Indeed, we can quite simply see that Reeb transforms do not distinguish all shapes: for instance, the sphere, the solid ball, and a spherical shell can yield identical Reeb graphs, since all their slices appear as single connected components regardless of whether these slices are disks, annuli, or circles. Figure~\ref{fig:cone} illustrates similar counterexamples: a solid cone and a hollow cone can also exhibit the same Reeb graph. 
%Hence, one may ask whether the Reeb transform---the collection of Reeb graphs across \emph{all} coordinate directions---is truly invertible.

\begin{figure}
    \centering
    \includegraphics[width=\linewidth]{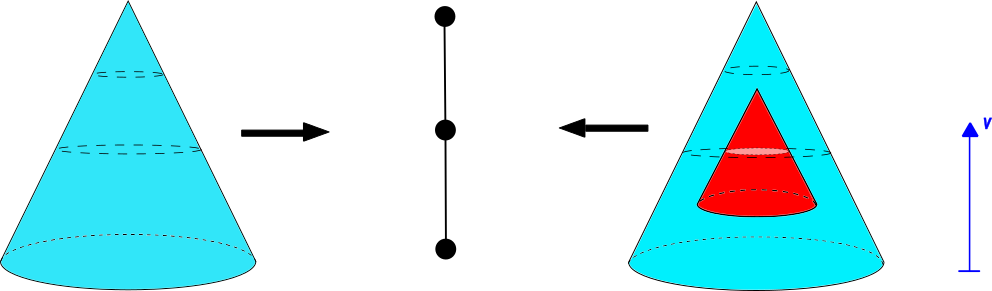}
    \caption{Figure illustrating that both a solid cone and a hollow cone have identical Reeb graphs, as each level set is a single component in both.}
    \label{fig:cone}
\end{figure}

While the Reeb transform is not universal in distinguishing arbitrary shapes, our research indicates it is, in fact, \emph{injective} for certain classes of objects. For example, shapes in low-dimensional settings can be reconstructed from their Reeb transforms under additional regularity assumptions. %Our work continues this investigation, presenting sub families of sets with positive injectivity results.
%and explicit failure cases in higher-dimensional domains.
%This question of “when does the Reeb transform invert a shape?” has motivated us to adopt an o-minimal perspective to avoid pathological behavior.
%In an o-minimal structure, definable sets enjoy finite stratifications, permitting Reeb graphs to remain finite combinatorial objects. This approach follows van den Dries~\cite{vandenDriesOminimal} and subsequent authors exploring tame topology~\cite{dwyer1979tame, FUJITA_2021, fujita2023locally}.

A critical aspect of studying Reeb transforms involves understanding how the connectivity of level sets changes. Various works~\cite{Curry2018, DBLP:journals/corr/abs-1912-12759, 1e1d4861e32c4b0c817b02fe130cd5a2} show that knowing the Euler characteristics of these level sets across all directions can suffice to reconstruct large classes of shapes. Our research leverages and extends these ideas to prove injectivity for surfaces in $\mathbb{R}^3$. Conversely, we provide examples demonstrating that, above three dimensions, the Reeb transform ceases to be injective, underscoring its limitations when more complex topological structures arise.
The end result is an interesting picture of Reeb transforms as an intermediate-level shape descriptor: highly effective in certain tame, low-dimensional settings, but inevitably losing information for more complex domains.

\section{Prelims and Definitions}
In this section, we establish the core concepts and notations needed throughout the paper. We begin with the notion of \emph{definable sets} in an o-minimal structure and then introduce the classes $Set(d)$, $Set(d,1)$, and $\Surf(3)$. Next, we define the height functions, level sets, and Reeb graphs. We discuss the Euler Characteristic Transform and its injectivity results which leads to the introduction of the Reeb transform, and an equivalence relation that classifies shapes based on their Reeb transforms. 

\subsection{Definable Sets, surfaces, and the Euler characteristic transform}

\begin{definition}[o‑minimal Structure]\label{def:o-minimal}
An \emph{o‑minimal structure} on the ordered field \((\mathbb{R},<,+,\cdot,0,1)\) is a sequence 
\(
\mathcal{S}=\bigl(\mathcal{S}_{n}\bigr)_{n\ge 1}
\)
such that for each positive integer \(n\)

\begin{enumerate}
\item \(\mathcal{S}_{n}\) is a Boolean algebra of subsets of \(\mathbb{R}^{n}\) (closed under finite unions, intersections, and complements);
\item if \(A\in\mathcal{S}_{m}\) and \(B\in\mathcal{S}_{n}\), then \(A\times B\in\mathcal{S}_{m+n}\);
\item if \(A\in\mathcal{S}_{n+1}\), the projection
      \(\pi(A)=\{\,x\in\mathbb{R}^{n}\mid\exists y\in\mathbb{R}:(x,y)\in A\}\)
      lies in \(\mathcal{S}_{n}\);
\item \(\mathcal{S}_{n}\) contains every algebraic subset of \(\R^{n}\)
      (i.e.\ every zero set of a real polynomial in \(n\) variables);
\end{enumerate}

A set $A\subset\mathbb{R}^n$ is called \emph{definable} (in $\mathcal{S}$) if $A\in\mathcal{S}_n$.
A function is \emph{definable} if its graph is a definable subset of $\mathbb{R}^{n+1}$.

The structure $\mathcal{S}$ is called \emph{o-minimal} if, in addition,
every set in $\mathcal{S}_1\subset \mathbb{R}$ is a finite disjoint union of open intervals and points.
\end{definition}
The o‑minimal axiom guarantees “tameness’’: definable subsets of \(\R^{n}\) admit finite cell decompositions, have finitely many connected components, and avoid pathological oscillations. 
The requirement of definability (in an o-minimal structure) ensures that sets have \emph{tame} topological properties, avoiding pathological cases and allowing for finite stratifications.

If we fix an ordered orthonormal basis of $\R^d$ we can define the corresponding projections onto the subspaces spanned by the first $j$ basis elements. For
$1\le j\le d$ we write $\pi_j:\R^d\to\R^j$ as this projection onto the first $j$ coordinates.
 For $D\subseteq\R^{j}$, and a function $\zeta:D\to\R$ we write $\Gamma(\zeta)=\{(x,\zeta(x)) \mid x\in D\}\subseteq\R^{j+1}$ for its graph.

\begin{definition}[Definable function]
Let $A\subset \R^j$ be definable. A function $\zeta:A\to \R$ is a \emph{definable function} if its graph $\Gamma(\zeta)$ us a definable set.   
\end{definition}

Some useful facts about definable functions are 
\begin{enumerate}
\item The image of a definable function is definable \cite[Proposition ~1.6]{coste1999introduction}
\item Any polynomial, and in particular linear, function is definable \cite[Exercise ~1.7]{coste1999introduction}
\end{enumerate}

Together these imply that if $A\subset \R^d$ is definable and $X\in O(d)$ is a change of basis then $X(A)\subset \R^d$ is also definable.

It will be useful to have cell decompositions that play well with the projection maps, called cylindrical cell decomposition.

\begin{definition}[Cylindrical definable cell decomposition]
\label{def:cdcd}
A \emph{cylindrical definable cell decomposition} (cdcd) of $\R^d$ of $\R^d$ is a finite partition $\mathcal{C}_d$ of $\R^d$ into definable sets, called \emph{cells}, satisfying a cylindrical condition defined inductively by dimension $d$. 

For $d=1$, a cdcd of $\R$ is determined by a finite (possibly empty)
set of points $x_1<x_2<\cdots<x_p$ in $\R$. Its cells are the
singletons $\{x_1\},\dots,\{x_p\}$ and the open intervals
\[
  (-\infty,x_1),\quad (x_1,x_2),\quad\dots,\quad (x_{p-1},x_p),
  \quad (x_p,+\infty).
\]
If $p=0$ the unique cell is $\R$ itself.

For $d>1$ a cdcd $\mathcal{C}_d$ of $\R^d$ consists of a cdcd
$\mathcal{C}_{d-1}$ of $\R^{d-1}$ together with, for each cell
$D\in\mathcal{C}_{d-1}$, an integer $m_D\ge 0$ and definable continuous functions
\[
  \zeta_{D,1}<\zeta_{D,2}<\cdots<\zeta_{D,m_D}\;:\;D\longrightarrow\R ,
\]
where the inequalities are required to hold pointwise on $D$. Setting
$\zeta_{D,0}\equiv-\infty$ and $\zeta_{D,m_D+1}\equiv+\infty$, the cells
of $\mathcal{D}$ lying over $D$ are the \emph{graphs}
\[
  \Gamma(\zeta_{D,l}), \qquad 1\le l\le m_D ,
\]
and the \emph{bands}
\[
  (\zeta_{D,l},\zeta_{D,l+1})
  \;=\;\bigl\{(x,y)\in D\times\R \;\bigm|\;
  \zeta_{D,l}(x)<y<\zeta_{D,l+1}(x)\bigr\},
  \qquad 0\le l\le m_D .
\]
The cells of $\mathcal{C}_d$ are all cells obtained in this way for $D\in \mathcal{C}_{d-1}$.
\end{definition}

By construction $\pi_j$ maps each cell of
$\mathcal{C}_d$ onto a cell of $\mathcal{C}_{j}$, for all $1\leq j<d$.

\begin{definition}[Compatible decompositions]\label{def:adapted}
A cdcd $\mathcal{C}_d$ of $\R^d$ is
\emph{compatible} with definable sets $X_1,\dots,X_m\subseteq\R^d$ if
each $X_j$ is a union of cells of $\mathcal{C}_d$.
\end{definition}

\begin{definition}[$C^k$ decompositions]\label{def:ckcdcd}
Let $k\ge 1$. A cdcd of $\R^d$ is a \emph{$C^k$ cdcd} if, at every level
of the inductive construction in Definition~\ref{def:cdcd}, all of the
functions $\zeta_{D,l}$ are $C^k$ on $D$.
\end{definition}

The fundamental result we will be using is the existence of cylindrical definable cell decompositions.

\begin{thm}[$C^k$ cell decomposition; {\cite[Thm.~6.6]{coste1999introduction}}]
\label{thm:celldecomp}
Let $k\ge 1$ and let $X_1,\dots,X_m\subseteq\R^d$ be definable. Then
there exists a $C^k$ cylindrical definable cell decomposition of $\R^d$
adapted to $X_1,\dots,X_m$. Moreover, given a definable function
$f:A\to\R$ with $A\subseteq\R^d$ definable, the decomposition may be
chosen adapted to $A$ and such that $f|_C$ is $C^k$ for every cell
$C\subseteq A$.
\end{thm}

A useful theorem for proving the existence of deformation retractions is Theorem 4.6 in \cite{coste1999introduction}.

\begin{prop}
\label{prop:4.6}
Let $A$ be a closed and bounded definable subset of $\R\times\R^{n}$ and let
$B_i$, $i=1,\dots,k$, be definable subsets of $A$. Let
$\pi:\R\times\R^{n}\to\R$ be the projection on the first coordinate. Then there
exist a finite simplicial complex $K$ with vertices in $\R\times\mathbb{Q}^{n}$ and a
definable homeomorphism $\Phi:|K|\to A$ such that
\[
  \pi\circ\Phi=\pi|_{|K|}
\]
and each $B_i$ is a union of images by $\Phi$ of open simplices of $K$.
\end{prop}

More details on o-minimal sets and cell decompositions  are available in \cite{coste1999introduction} and  \cite{vandenDriesOminimal}. In this paper we will restrict our attention to particular families of compact definable sets.

\begin{definition}%[Definable Sets in an O-Minimal Structure]
\label{def:S(n)}
We say that $A\subset \R^d$ is \emph{constructible} if it is both definable and compact.
Let $Set(d)$ denote the class of all constructible subsets of $\mathbb{R}^d$.
\end{definition}

We will see that although the Reeb Transform will not be injective when the domain is the entirety of $\Set(d)$, we do get injectivity results when we restrict to certain families of constructible sets. These are the constructible sets that decompose into $0$ and $1$ dimensional cells, and those in $\R^3$ that are homeomorphic to surfaces. 

\begin{definition}\label{def:Set(d,1)}
We write $\Set(d,1)\subset \Set(d)$ for those $A$ that admit a definable stratification
\[
A \;=\; A_0 \sqcup A_1,
\]
where $\dim(A_0)=0$ and $\dim(A_1)=1$, and $A_1$ is a (possibly disconnected) definable $C^1$
embedded $1$--dimensional submanifold of $\R^d$.
\end{definition}

\begin{definition}\label{def:Surf(3)}
We write $\Surf(3)\subset \Set(3)$ for the class of \emph{(embedded) closed definable surfaces} in
$\R^3$: namely, $A\in\Surf(3)$ if and only if $A\subset \R^3$ is compact and definable and is a
homeomorphic to a $2$--dimensional submanifold of $\R^3$ \emph{without boundary}.
\end{definition}

The Euler characteristic transform of a shape in $\mathbb{R}^d$, defined in \cite{Curry2018}, consists of the Euler characteristic function of all sublevel sets for all directions in $S^{d-1}$. 

\begin{definition}
    Let $A \in Set(d)$. The Euler Characteristic Transform of $A$ is the function $ECT(A): S^{d-1} \to \mathbb{Z}^{\mathbb{R}}$ (where $\mathbb{Z}^{\mathbb{R}}$ denotes the set of integer-valued functions on $\mathbb{R}$) defined by:
    \[ECT(A)(v)(t)=\chi(\{x\in A\mid x\cdot v\leq t\}).\]
\end{definition}

We briefly mention the following result, which follows from \cite[Theorem~3.5]{Curry2018} and its proof.

\begin{thm}(\cite[Theorem~3.5]{Curry2018})
\label{thm:ect-inj}
Let $X,Y\in \Set(d)$.
\begin{enumerate}
\item If
\[
\chi(\{x\in X \mid x\cdot v = t\})=\chi(\{y\in Y \mid y\cdot v = t\})
\quad\text{for all } v\in S^{d-1}\text{ and all }t\in\R,
\]
then $X=Y$ as subsets of $\R^d$.
\item If
\[
\chi(\{x\in X \mid x\cdot v \le t\})=\chi(\{y\in Y \mid y\cdot v \le t\})
\quad\text{for all } v\in S^{d-1}\text{ and all }t\in\R,
\]
then $X=Y$ as subsets of $\R^d$.
\end{enumerate}
\end{thm}

We can make this statement even stronger, showing agreement over a dense subset of directions is sufficient. To do this we will use some properties and consider Euler curves as elements of the $L^1$ space of functions. If $\chi(A)\neq 0$ then the raw Euler curves are not integrable. An easy workaround is to  restrict the domain to $[-R,R]$ where $R>0$ is chosen large enough so $A\subset B(0, R)$. Consider the function space $L^1([-R, R])$, consisting of functions $f$ such that $\|f\|_1=\int_{-R}^R |f(t)|\, dt$ is finite. Then natural metric on this space is $d_1(f,g)=\|f-g\|_1$. As $A\subset B(0, R)$ we have $\ECT(A)(v)(t)=\chi(A)$ for all $t>R$ and $\ECT(A)(v)(t)=0$ for all $t<R$. As such $\int_{-\infty}^\infty |\ECT(A)(v)(t)-\ECT(A)(w)(t)|\, dt=\int_{-R}^R |\ECT(A)(v)(t)-\ECT(A)(w)(t)|\, dt$, showing that any distance between Euler curves is independent of the choice of $R$.

\begin{prop}\label{prop:dense}
   Let $A\in \Set(d)$. We can determine $\ECT(A)$ from $\{(v, \ECT(A)(v))\mid v\in U\}$ whenever $U$ is a dense subset of $S^{d-1}$. 
\end{prop}
\begin{proof}
     We know from \cite[Prop.\ 4.13]{Curry2018} that for any fixed $A\in \Set(d)$, its Euler Characteristic Transform $\ECT(A)$ is a continuous function from $S^{d-1}$ to $L^1([-R,R])$. This result is not surprising as it is highly analogous to standard persistent homology stability theorems, observing that  $\|h_v-h_w\|_\infty \leq \|v-w\|_2$. 

 Suppose that we know $\{\ECT(A)(v_i)\}$ for a convergence sequence of directions $v_i \to \hat{v}\in S^{d-1}$. Using the continuity we know that $\ECT(A)(\hat{v})=\lim_{i\to \infty} \ECT(A)(v_i)$ considered as an element of $L^1(\R)$. It is worth observing that this result on its own only determines $\ECT(A)(\hat{v})$ up to a measure zero subset of the domain $\R$.

The Euler characteristic curves of height functions over compact subsets is right-continuous (see \cite{ji2023euler}). For each element in $L^1([-R,R])$ there is a unique representative which is right continuous. This means we thus know $\ECT(A)(\hat{v})$ exactly, not just up to a set of measure zero.

If $U$ is a dense subset then for all $\hat{v}\in S^{d-1}$ there exists a convergent sequence $v_i \to \hat{v}$ with each $v_i\in U$. By this process we can determine $\ECT(A)(v)$ for all $v\in S^{d-1}$
\end{proof}

We can combine Proposition \ref{prop:dense} with \ref{thm:ect-inj} we get the stronger corollary that whenever the Euler characteristic curves agree on a dense subset then the original subsets are equal. 

\begin{cor}\label{cor:inj_dense}
Let $A_1,A_2\in \Set(d)$. If $\ECT(A_1)(v)=\ECT(A_2)(v)$ for all $v\in U$, $U$ a dense subset of $S^{d-1}$ then $A_1=A_2$ as subsets of $\R^d$.    
\end{cor}

\subsection{Introducing the Reeb Transform}

\begin{definition}[Level Sets]
Let $A\in Set(d)$ and $v\in S^{d-1}$. Define the height function
\[
h_v:A\to\R,\qquad h_v(x):=x\cdot v.
\]
For $t\in\R$, the level set of $A$ in direction $v$ at height $t$ is
\[
LS(A,v,t)
:=\{x\in A\mid h_v(x)=t\}
=\{x\in A\mid x\cdot v=t\}.
\]
\end{definition}

A topological space $X$ is said to be \emph{connected} if it cannot be partitioned into two non-empty, disjoint, open subsets. We define a relation $\sim$ on $Y\in Set(d)$ by
\[
x \sim y \quad\Longleftrightarrow\quad \exists\ \text{a connected subset }C\subseteq X\text{ such that }x,y \in C.
\]
We define an equivalence relation on $A$ by declaring $x\sim y$ if $x$ and $y$ lie in the same
connected component of the level set $LS(A,v,h_v(x))$.

\begin{definition}[Reeb quotient space]\label{def:ReebQuotient}
Let $A \in \Set(d)$ and $v \in S^{d-1}$. Define an equivalence relation $\sim_v$ on $A$ by
declaring that for $x,y \in A$,
\[
x \sim_v y
\quad\Longleftrightarrow\quad
v\cdot x = v\cdot y
\ \text{and}\ 
x,y \text{ lie in the same connected component of } LS(A,v,v\cdot x).
\]
The \emph{Reeb quotient space} of $A$ in direction $v$ is the quotient topological space
\[
\mathcal R_v(A) \ :=\ A/\!\sim_v,
\]
with quotient map $q_v:A\to \mathcal R_v(A)$. The induced height function
$\tilde h_v:\mathcal R_v(A)\to \R$ is defined by $\tilde h_v([x])=v\cdot x$.
\end{definition}

For definable sets, $\mathcal R_v(A)$ is a definable space and is realizable as a finite graph-like
(one-dimensional) object; see~\cite{basu2022reeb}. In this paper we will only require this finiteness
and graph structure for the classes on which we prove injectivity (in particular, subsets of $\R^2$,
one-dimensional embedded complexes, and compact surfaces in $\R^3$).

\begin{definition}[Reeb graph]\label{def:ReebGraph}
Let $A$ and $v$ be as above. The \emph{Reeb graph} of $A$ in direction $v$ is the quotient space
$\mathcal R_v(A)$ from Definition~\ref{def:ReebQuotient}, equipped with the following canonical
graph structure and the function $\tilde h_v$.

For a point $p\in \mathcal R_v(A)$, define its \emph{local degree} $\deg(p)$ to be the number of
connected components of $U\setminus\{p\}$, where $U$ is any sufficiently small neighbourhood of $p$
in $\mathcal R_v(A)$. (For a finite topological graph this number is independent of the choice of
sufficiently small $U$.)

Define the vertex set
\[
V \ :=\ \{\,p\in \mathcal R_v(A)\mid \deg(p)\neq 2\,\},
\]
and define the edge set $E$ as the set of connected components of $\mathcal R_v(A)\setminus V$.
Each edge is homeomorphic to an open interval, and $\tilde h_v$ restricts on each edge to a
homeomorphism to an open interval in $\R$ (possibly unbounded).
\end{definition}

Previous work (e.g.,~\cite{Bauer2020}) defines isomorphisms of Reeb graphs as piecewise-linear
homeomorphisms preserving the associated Reeb functions. In our setting, this is equivalent to the
following formulation in terms of graph isomorphisms, together with preservation of height values.

\begin{definition}[Isomorphism of Reeb graphs]\label{def:ReebIso}
Let $A,B\in\Set(d)$ and let $v,w\in S^{d-1}$. Let $\bigl(\mathcal R_v(A),\tilde h_v\bigr)$ and
$\bigl(\mathcal R_w(B),\tilde h_w\bigr)$ be the associated Reeb graphs as in
Definition~\ref{def:ReebGraph}. We say that $\mathcal R_v(A)$ and $\mathcal R_w(B)$ are
\emph{isomorphic}, and write $\mathcal R_v(A)\cong \mathcal R_w(B)$, if there exists a homeomorphism
\[
\Phi:\mathcal R_v(A)\longrightarrow \mathcal R_w(B)
\]
such that
\[
\tilde h_w\circ \Phi \;=\; \tilde h_v
\qquad\text{on all of }\mathcal R_v(A).
\]
\end{definition}

\begin{definition}[Reeb Transform]\label{def:ReebTransform}
Let $A\in\Set(d)$. The \emph{Reeb Transform} of $A$ is the map
\[
\mathcal{RT}(A): S^{d-1}\longrightarrow \mathcal{G},
\qquad
\mathcal{RT}(A)(v)\ :=\ \bigl(\mathcal{R}_{v}(A),\tilde h_v\bigr),
\]
where $\mathcal{G}$ denotes the collection of Reeb graphs equipped with their induced height
functions (up to the isomorphism of Definition~\ref{def:ReebIso}).
\end{definition}

It is both theoretically and practically useful to regard $\mathcal{RT}(A)$ as a map from the unit
sphere into the metric space of Reeb graphs. This viewpoint motivates approximating the Reeb
Transform by evaluating $\mathcal{RT}(A)(v)$ on a finite sample of directions $v\in S^{d-1}$.
In proposition~\ref{prop:single_dir_stability}, we quantify the approximation error using a directional stability bound for Reeb graphs.

There are many different potential metrics to compare Reeb graphs; see~\cite{Bollen2021} for a recent survey of methods and their tradeoffs.
For the purposes of this paper, we consider functional distortion distance~\cite{Bauer2014}.  
Borrowing ideas from the more well-known Gromov-Hausdorff distance, the functional distortion distance considers all possible maps between the two graphs and defines the distortion of a pair of such maps as the largest change in function values between any path in one graph and its image in the other graph.
The functional distortion distance is then the infimum of the distortion over all possible maps between the two graphs. More formally:

\begin{definition}[Functional distortion distance {\cite{bauer2015strong}}]\label{def:fd}
Let $(R_f,f)$ and $(R_g,g)$ be Reeb graphs equipped with their Reeb functions
$f:R_f\to \R$ and $g:R_g\to \R$.
For $x,x'\in R_f$, define
\[
d_f(x,x') \ :=\ \inf_{\gamma}\Big(\max_{t\in[0,1]} f(\gamma(t))-\min_{t\in[0,1]} f(\gamma(t))\Big),
\]
where the infimum ranges over all continuous paths $\gamma:[0,1]\to R_f$ with
$\gamma(0)=x$ and $\gamma(1)=x'$. Define $d_g$ analogously on $R_g$.

Given continuous maps $\phi:R_f\to R_g$ and $\psi:R_g\to R_f$, set
\[
G(\phi,\psi)
=\{(x,\phi(x)):\ x\in R_f\}\ \cup\ \{(\psi(y),y):\ y\in R_g\},
\]
and
\[
D(\phi,\psi)
=\sup_{(x,y),(x',y')\in G(\phi,\psi)}\frac12\,\bigl|d_f(x,x')-d_g(y,y')\bigr|.
\]
The \emph{functional distortion distance} between $(R_f,f)$ and $(R_g,g)$ is
\[
d_{\mathrm{FD}}\bigl((R_f,f),(R_g,g)\bigr)
=\inf_{\phi,\psi}\max\Big\{
D(\phi,\psi),\ \|f-g\circ\phi\|_\infty,\ \|f\circ\psi-g\|_\infty
\Big\},
\]
where the infimum ranges over all continuous maps $\phi:R_f\to R_g$ and $\psi:R_g\to R_f$ (see figure \ref{fig:functional-distortion-example}).
\end{definition}

\begin{figure}[t]
    \centering

    \begin{minipage}[t]{0.48\textwidth}
        \centering
        \includegraphics[width=\linewidth]{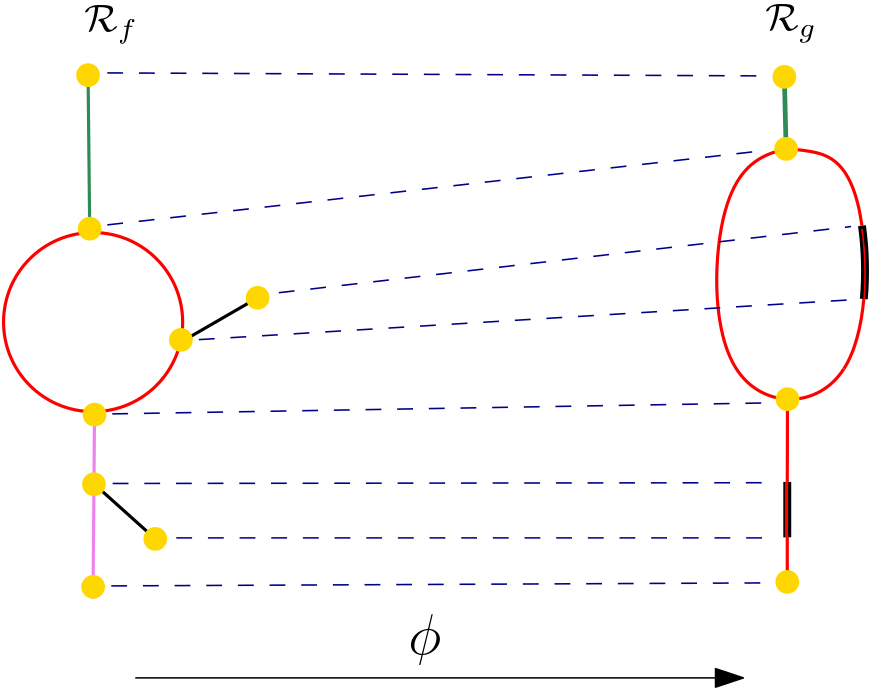}
        
        \smallskip
        \textbf{(a)} 
        $\Phi:\mathcal{R}_f\to\mathcal{R}_g$ maps the additional tail of
        $\mathcal{R}_f$ into the corresponding portion of the loop in
        $\mathcal{R}_g$.
    \end{minipage}
    \hfill
    \begin{minipage}[t]{0.48\textwidth}
        \centering
        \includegraphics[width=\linewidth]{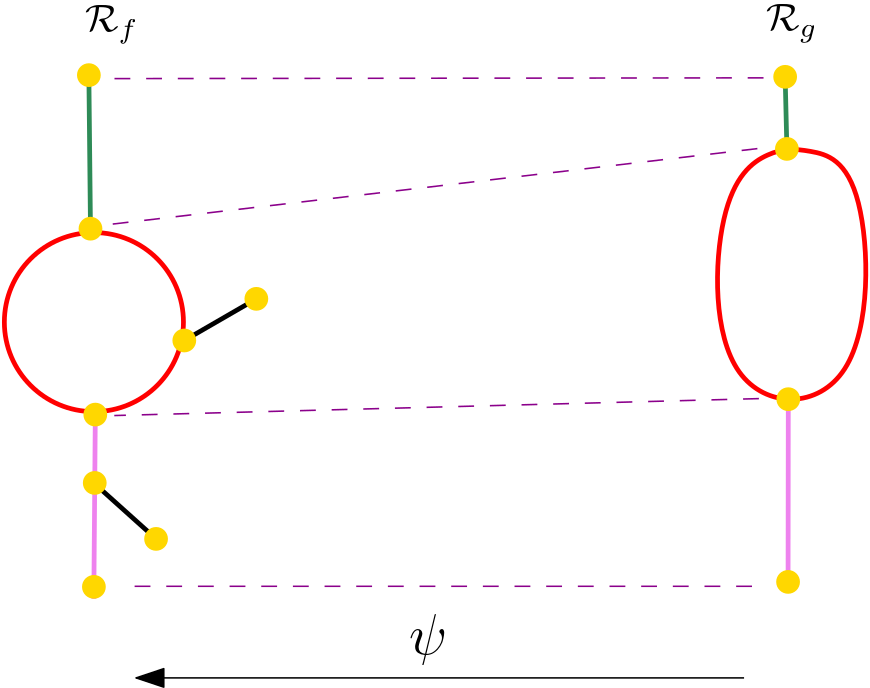}
        
        \smallskip
        \textbf{(b)}
        $\Psi:\mathcal{R}_g\to\mathcal{R}_f$ maps into the common
        trunk-and-loop structure, leaving the additional tail of
        $\mathcal{R}_f$ outside its image.
    \end{minipage}

    \caption{
    Two Reeb graphs $\mathcal{R}_f$ and $\mathcal{R}_g$, together with
    continuous maps $\Phi:\mathcal{R}_f\to\mathcal{R}_g$ and
    $\Psi:\mathcal{R}_g\to\mathcal{R}_f$ illustrating the functional
    distortion distance. The graphs share the same trunk-and-loop structure,
    while $\mathcal{R}_f$ contains an additional tail attached to the loop.
    Dotted lines indicate selected point correspondences induced by the maps,
    and matching colors indicate paths in one graph together with their images
    in the other. The mismatch created by the additional tail is the primary
    source of distortion for the illustrated pair of maps.
    }
    \label{fig:functional-distortion-example}
\end{figure}

\begin{prop}[Directional stability of Reeb graphs]\label{prop:single_dir_stability}
Let $A\in Set(d)$ be non-empty, and let $D>0$ be such that
\[
A\subset \overline{B}(0,D)
:=\{x\in\R^d:\|x\|_2\le D\},
\]
where $\overline{B}(0,D)$ denotes the closed Euclidean ball of radius $D$ centered at the origin.
Fix $v,w\in S^{d-1}$. Then
\[
d_{\mathrm{FD}}\!\Bigl(\bigl(\mathcal R_v(A),\tilde h_v\bigr),\ \bigl(\mathcal R_w(A),\tilde h_w\bigr)\Bigr)
\;\le\; D\,\|v-w\|_2
\;=\; 2D\,\sin\!\Bigl(\frac{\theta}{2}\Bigr),
\]
where $\theta=\arccos\langle v,w\rangle$ and $d_{\mathrm{FD}}$ is the functional distortion distance.
In particular, for small $\theta$ this bound is $O(D\,\theta)$.
\end{prop}

\begin{proof}
Let $h_v,h_w:A\to \R$ be given by $h_v(x)=v\cdot x$ and $h_w(x)=w\cdot x$.
For any $x\in A$,
\[
|h_v(x)-h_w(x)|
=|(v-w)\cdot x|
\le \|v-w\|_2\,\|x\|_2
\le D\,\|v-w\|_2.
\]
Hence $\|h_v-h_w\|_\infty\le D\,\|v-w\|_2$.

By~\cite[Thm.\,4.1]{Bauer2014}, the Reeb graphs of $h_v$ and $h_w$ satisfy
\[
d_{\mathrm{FD}}\bigl((R_{h_v},\bar h_v),(R_{h_w},\bar h_w)\bigr)\ \le\ \|h_v-h_w\|_\infty.
\]
In our notation, $(R_{h_v},\bar h_v)=(\mathcal R_v(A),\tilde h_v)$ and
$(R_{h_w},\bar h_w)=(\mathcal R_w(A),\tilde h_w)$, so the claimed inequality follows.
Finally, since $\|v-w\|_2=\sqrt{2-2\langle v,w\rangle}=2\sin(\theta/2)$, we obtain the stated
trigonometric form.
\end{proof}

\begin{remark}\label{rem:lip}
The map
\[
v \longmapsto \bigl(\mathcal R_v(A),\tilde h_v\bigr)
\]
from $S^{d-1}$ to the space of Reeb graphs equipped with Reeb functions is $D$--Lipschitz with
respect to $d_{\mathrm{FD}}$ and the ambient Euclidean metric on $S^{d-1}$. Equivalently, rotating
the slicing direction by an angle $\theta$ perturbs the Reeb graph by at most
$2D\sin(\theta/2)$ in functional distortion distance.
\end{remark}

\begin{cor}[Directional stability of the Reeb Transform]\label{cor:rt_dir_stability}
Let $A\in \Set(d)$ be non-empty and assume $A\subset B(0,D)$.
Then the Reeb Transform
\[
\mathcal{RT}(A): S^{d-1}\to \mathcal{G},\qquad \mathcal{RT}(A)(v)=(\mathcal R_v(A),\tilde h_v),
\]
is $D$--Lipschitz with respect to the Euclidean metric on $S^{d-1}$ and the functional distortion
distance on $\mathcal{G}$. In particular, for all $v,w\in S^{d-1}$,
\[
d_{\mathrm{FD}}\!\bigl(\mathcal{RT}(A)(v),\mathcal{RT}(A)(w)\bigr)
\le D\|v-w\|_2
=2D\sin\!\Bigl(\frac{\theta}{2}\Bigr),
\quad \theta=\arccos\langle v,w\rangle.
\]
\end{cor}

\begin{proof}
This is exactly Proposition~\ref{prop:single_dir_stability}, rewritten using the definition of
$\mathcal{RT}(A)$.
\end{proof}

We can use the notion of equivalence between Reeb graphs to define an equivalence relation on Reeb Transforms.

\begin{definition}[Equivalence under the Reeb Transform]\label{def:equiv-simR}
We define an equivalence relation $\sim^{\mathcal{RT}}$ on $\Set(d)$ as follows: for $A,B\in\Set(d)$,
\[
A \sim^{\mathcal{RT}} B
\quad\Longleftrightarrow\quad
\text{for every } v\in S^{d-1},\ 
\bigl(\mathcal R_v(A),\tilde h_v\bigr)\cong \bigl(\mathcal R_v(B),\tilde h_v\bigr),
\]
where $\cong$ denotes isomorphism of Reeb graphs in the sense of Definition~\ref{def:ReebIso}.
\end{definition}

Thus, $A$ and $B$ are equivalent under $\sim^{\mathcal{RT}}$ precisely when their Reeb graphs
agree (up to isomorphism) in every slicing direction.

Unlike the Euler Characteristic Transform, distinct constructible sets may have equivalent Reeb
Transforms. Assume $d\ge 3$ and consider
\[
A=\{x\in\R^d:\ \|x\|\le 2\}\qquad\text{and}\qquad
B=\{x\in\R^d:\ 1\le \|x\|\le 2\}.
\]
Then $A\neq B$; moreover $A$ is contractible whereas $B$ is not.
Fix any $v\in S^{d-1}$. For $|t|>2$ the level set $LS(A,v,t)$ (and $LS(B,v,t)$) is empty, and for
$|t|\le 2$ each level set is connected. Consequently, for every $v$ the Reeb graphs
$\bigl(\mathcal R_v(A),\tilde h_v\bigr)$ and $\bigl(\mathcal R_v(B),\tilde h_v\bigr)$ are both
isomorphic to a single interval with height parameter ranging from $-2$ to $2$.
Hence $A\sim^{\mathcal{RT}} B$ even though $A\neq B$.

This example illustrates that the Reeb Transform records the evolution of connected components of
level sets across directions, but it does not detect the topology within individual components.
Therefore, the Reeb Transform is not injective on the full class $\Set(d)$.

\section{Positive Injectivity Results}

%Firstly we will cover the simple case of  subsets of $\R^2$, where the Reeb transform does completely encode the embedded shape.

\begin{figure}
    \centering
    
\includegraphics[width=\textwidth]{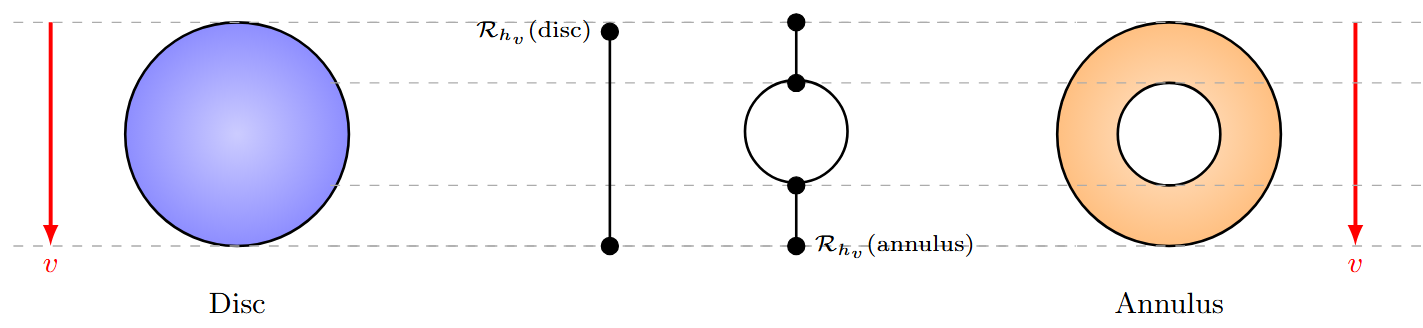}
\caption{Example of two different surfaces in $\mathbb{R}^2$ (a disc and an annulus) and a direction where their Reeb graphs differ. Such a direction must exist for any pair of distinct constructible subsets in $\R^2$}
\label{Fig: S2}
\end{figure}

\begin{lemma}\label{lem:S(2)_inj}
The Reeb Transform is injective when restricted to $\Set(2)$.
\end{lemma}

\begin{proof}
By Theorem~\ref{thm:ect-inj}, it suffices to show that for each $A\in\Set(2)$, each direction
$v\in S^{1}$, and each $t\in\R$, the Reeb graph $(\mathcal R_v(A),\tilde h_v)$ determines
$\chi(LS(A,v,t))$.

Fix $v\in S^{1}$ and $t\in\R$(see Figure~\ref{Fig: S2}). Since $A\in\Set(2)$, the level set $LS(A,v,t)$ is a finite disjoint
union of points and closed intervals. In particular, every connected component of $LS(A,v,t)$ has
Euler characteristic~$1$, and hence
\[
\chi(LS(A,v,t)) \;=\; \#\{\text{connected components of }LS(A,v,t)\}.
\]

Let $q_v:A\to \mathcal R_v(A)$ be the quotient map. By construction of the Reeb quotient,
the connected components of $LS(A,v,t)$ are precisely the equivalence classes of $\sim_v$ at height $t$,
and therefore are in bijection with the fiber $\tilde h_v^{-1}(t)\subset \mathcal R_v(A)$.
Consequently,
\[
\chi(LS(A,v,t)) \;=\; \bigl|\tilde h_v^{-1}(t)\bigr|,
\]
which is determined by $(\mathcal R_v(A),\tilde h_v)$.
\end{proof}

%records the number of connected components of the intersection 
%$L\cap A$ for \emph{every} line $L$ with direction~$v$.

%%Scholium~\ref{scol:EC} 
%Because $A\subset\R^{2}$ is compact and definable we can apply the theory in \ref{ECT}. This showed that the

%Euler characteristics
%$\chi(L\cap A)$ for \emph{all} lines and \emph{all} directions uniquely
%determines a set $A\in\Set(2)$.  
%In the planar setting $\chi(L\cap A)$ equals the number of connected
%components of $L\cap X$, so $\RT(A)$ already provides this complete data needed to 
%for $X=A$.

% \begin{comment}
% Now suppose $\RT(A)=\RT(B)$ for two sets $A,B\in\Set(2)$.  
% Then for every line $L$ we have the same Reeb graph in direction~$v$,
% and therefore $\chi(L\cap A)=\chi(L\cap B)$.  
% By Scholium~\ref{scol:EC} the two sets must coincide, i.e.\
% $A=B$.  Consequently the Reeb transform is injective on $\Set(2)$.
% \end{comment}

Now we restrict to $1$--dimensional stratified spaces, i.e., definable subsets of $\R^d$ whose
strata have dimension at most $1$. Equivalently, these are (finite) graphs embedded in $\R^d$ with
definable image. We denote this class by $\Set(d,1)$.

%For $A\in \Set(d,1)$ and a fixed direction $v\in S^{d-1}$, the level set
%$LS(A,v,t)=\{x\in A\mid x\cdot v=t\}$ is typically a finite set of points for each $t\in\R$.
%In this generic situation, $\chi(LS(A,v,t))$ is simply the number of points in the fiber
%$\tilde h_v^{-1}(t)\subset \mathcal R_v(A)$.
For a fixed $A\in\Set(d,1)$ and a direction $v\in S^{d-1}$, the level set
$LS(A,v,t)=\{x\in A\mid x\cdot v=t\}$ is generically a finite set of contractible sets (often just single points) for each $t\in\R$.
In this situation, $\chi(LS(A,v,t))$ is simply the number of connected components of $LS(A,v,t)$, or
equivalently the number of points Reeb graph at height $t$. For non-contractible connected components to appear in a level set it is necessary for the height function is constant along a $1$--dimensional stratum,
so that an entire edge of $A$ lies in a single level set. To isolate these degenerate cases we introduce
the following notion.

\begin{definition}[Collapsing direction]\label{def:collapsing}
Let $A\in \Set(d,1)$. We call $v\in S^{d-1}$ a \emph{collapsing direction} for $A$ if there exists
$t\in\R$ such that $LS(A,v,t)$ contains an entire $1$--dimensional stratum of $A$.
\end{definition}

%We now restrict to $1$--dimensional stratified spaces in $\Set(d)$, i.e., definable subsets of $\R^d$
%whose strata have dimension at most $1$. We denote this class by $\Set(d,1)$.
%Thus an element $A\in\Set(d,1)$ decomposes into finitely many $0$--dimensional strata (vertices)
%and finitely many $1$--dimensional strata (edges).

%The only degeneracy we must exclude is when the height function is constant along an entire
%$1$--dimensional stratum, so that an entire edge of $A$ lies in a single level set. 
We want to isolate this phenomenon of collapsing directions. To do this we will use cylindrical definable cell decompositions.

\begin{lemma}\label{lem:collapsing}
Let $A\in \Set(d,1)$. The set of collapsing directions of $A$ has measure zero in $S^{d-1}$.
\end{lemma}

\begin{proof}
Let $\mathcal{C}$ be a $C^1$ cylindrical definable cell decomposition adapted to $A$.
Subdividing the edges if needed, each of the edges in $A$ are a cell in $\mathcal{C}$. This implies each (subdivided part of an) $1$--dimensional stratum
is a definable $C^1$ embedded curve; in particular, there exists a definable $C^1$ embedding
$\gamma_i:(0,\ell)\to \R^d$ whose image is $e_i$ and such that $\gamma_i'(s)\neq 0$ for all $s\in(0,\ell)$.
Choose any $s_0\in(0,\ell)$ and set $w=\gamma_i'(s_0)\neq 0$.

Suppose that $v\in S^{d-1}$ is such that $e_i\subset LS(A,v,t)$ for some $t\in\R$.
Then the function $h(s)=v\cdot \gamma_i(s)$ is constant on $(0,\ell)$, hence $h'(s)=0$ for all $s$.
In particular,
\[
0=h'(s_0)=v\cdot \gamma_i'(s_0)=v\cdot w.
\]
Therefore, any direction $v$ for which the entire edge $e_i$ is contained in a level set must lie in
\[
\{v\in S^{d-1}\mid v\cdot w=0\},
\]
which is a great $(d-2)$--sphere and thus has measure zero in $S^{d-1}$.

Finally, the set of collapsing directions for $A$ is contained in the union of these measure-zero sets
over the finitely many edges $e_1,\dots,e_k$, and hence has measure zero.
\end{proof}

\begin{lemma}\label{lem:Setd1_inj}
The Reeb Transform is injective when restricted to $\Set(d,1)$.
\end{lemma}

\begin{proof}
By Theorem~\ref{thm:ect-inj}, it suffices to show that
\[
\mathcal{RT}(X)\cong \mathcal{RT}(Y)\quad\Longrightarrow\quad \ECT(X)=\ECT(Y).
\]

Let $C_X$ and $C_Y$ denote the sets of collapsing directions for $X$ and $Y$
(Definition~\ref{def:collapsing}). By Lemma~\ref{lem:collapsing}, the set
$C_X\cup C_Y\subset S^{d-1}$ has measure zero, hence its complement is dense in $S^{d-1}$.

Fix $v\in S^{d-1}\setminus (C_X\cup C_Y)$. Then no $1$--dimensional stratum of $X$ (or $Y$) lies in a
level set of $h_v$. Consider a connected component $C$ of the level set $LS(X,v,t)$. Suppose that $C$ does not contain any $0$-strata of $X$. This then implies that $C$ is a connected subset of some $1$-strata and thus contractible. If $C$ contains a $0$-strata $x$ then $C$ is a connected subset of $X$ containing $x$ but each of the $1$ strata attached to $x$ can only be partially contained in $C$. Thus no other $0$-strata are in $C$ and $C$ must be homotopic to a star graph which is contractible.

As the fibers are contractible, the Reeb quotient maps
\[
q_v^X:X\to \mathcal R_v(X)
\qquad\text{and}\qquad
q_v^Y:Y\to \mathcal R_v(Y)
\]
are homotopy equivalences. Moreover, by definition of $\tilde h_v$ (Definition~\ref{def:ReebQuotient}),
they satisfy
\[
\tilde h_v^X\circ q_v^X \;=\; h_v|_X
\qquad\text{and}\qquad
\tilde h_v^Y\circ q_v^Y \;=\; h_v|_Y.
\]

Since $\mathcal{RT}(X)\cong \mathcal{RT}(Y)$, for this direction $v$ there exists an isomorphism of
Reeb graphs in the sense of Definition~\ref{def:ReebIso}, i.e., a homeomorphism
\[
\Phi_v:\mathcal R_v(X)\longrightarrow \mathcal R_v(Y)
\]
such that $\tilde h_v^Y\circ \Phi_v=\tilde h_v^X$. Transporting along $q_v^X$ and $q_v^Y$, define
\[
\Psi_v := (q_v^Y)^{-1}\circ \Phi_v\circ q_v^X : X\longrightarrow Y.
\]
Then $\Psi_v$ is a homeomorphism and
\[
h_v|_Y\circ \Psi_v
=\tilde h_v^Y\circ q_v^Y\circ (q_v^Y)^{-1}\circ \Phi_v\circ q_v^X
=\tilde h_v^Y\circ \Phi_v\circ q_v^X
=\tilde h_v^X\circ q_v^X
=h_v|_X.
\]
In particular, for every $t\in\R$ we have
\[
\Psi_v\!\bigl(\{x\in X:\ x\cdot v\le t\}\bigr)=\{y\in Y:\ y\cdot v\le t\},
\]
so the sublevel sets are homeomorphic and hence have the same Euler characteristic. Therefore,
\[
\ECT(X)(v,\cdot)=\ECT(Y)(v,\cdot)
\qquad\text{for all }v\in S^{d-1}\setminus (C_X\cup C_Y).
\]

Finally, Corollary \ref{cor:inj_dense} yields $X=Y$.
\end{proof}

We now turn to the case of compact definable surfaces in $\R^3$. Recall that $\Surf(3)\subset \Set(3)$ denotes the class of sets homeomorphic to a compact surface.
For $A\in \Surf(3)$, our goal is to relate the Euler characteristic of the planar slice $A\cap\{x\cdot v=t\}$ to combinatorial data of the Reeb graph $(\mathcal R_v(A),\tilde h_v)$, specifically the local degrees of its vertices at height $t$. Before doing this we first recall a useful fact about (sufficiently thin) thickened slices retracting to level sets.

\begin{lemma}\label{lem:retract}
    % Let $A\in \Surf(3)$, $v\in S^{2}$ and $t\in \R$ such that $LS(A,v,t)\neq \emptyset$. Then there is an $\epsilon>0$ such that $\{x\in A\mid x\cdot v \in [t-\epsilon, t+\epsilon]\}$ deformation retracts to $LS(A,v,t)$ along $f$ such that 
    % $$f^{-1}[0,s]=\{x\in A\mid x\cdot v \in [t-s\epsilon, t+s\epsilon]\}.$$ 

    Let $A\in Set(d)$, $v\in S^{d-1}$, and $t\in\R$ such that $LS(A,v,t)\neq\emptyset$, and define
    \[
    f:A\to\R_{\geq 0},\qquad f(x):=|h_v(x)-t|.
    \]
    Then there exists $\epsilon>0$ such that 
    \[
    \{x\in A\mid x\cdot v\in[t-\epsilon,t+\epsilon]\}
    =
    f^{-1}([0,\epsilon])
    \]
    deformation retracts to $LS(A,v,t)=f^{-1}(0)$.

\end{lemma}
\begin{proof}
We start our proof by observing that it is sufficient to prove the case when $v=e_1$. This is because we can construct  a change of basis matrix $X$ such that $Xv=e_1$. From $A\in\Set(d)$ and $X$ invertible we have $XA\in \Set(d)$. From the deformation retract computed for $X(A)$ in direction $e_1$, we can then construct the appropriate one for $v=e_1$ on $A$. This holds as $$\|h_v(x)-t\|=\|h_v(x-tv)\|=\|h_{e_1}(Xx-X(tv)))\|=\|h_{e_1}(Xx)-t\|.$$

Let $\pi$ denote the projection onto the first coordinate. For the case $v=e_1$ we have $f(x)=|\pi(x)-t|$. By Proposition \ref{prop:4.6} there
exists a finite simplicial complex $K$ with vertices in $\R\times\mathbb{Q}^{n}$ and a
definable homeomorphism $\Phi:|K|\to A$ such that
$\pi\circ\Phi=\pi|_{|K|}.$

Choose $\epsilon>0$ such that every vertex of $a\in K$ with $\pi|_{|K|}(a)\neq t$ has $\|\pi|_{|K|}(a)-t\|>\epsilon$.

Let $x\in K$ with $\pi|_{|K|}(x)\in [t-\epsilon, t+\epsilon]$. Let $\sigma=[a_0,a_1, \ldots a_m]$ be the smallest simplex containing $x$ and $a_0, a_1, \ldots a_k$ the vertices of $\sigma$ lying in $\pi_{|K|}^{-1}(t)$ (reordering the index as needed).  We can uniquely write $x=\sum_{i=0}^m \lambda_i a_i$ for some $\lambda_i\in [0,1]$. Set $\alpha:=\sum_{i=0}^k \lambda_i$.
Define $$h^K(x,s)=\sum_{i=0}^k \frac{\lambda_i}{\alpha+(1-s)(1-\alpha)} a_i +\sum_{i=k+1}^m\frac{(1-s)\lambda_i}{\alpha+(1-s)(1-\alpha)} a_i,$$ 
which is the straight line from $x$ to $\sum_{i=0}^k \frac{\lambda_i}{\alpha} a_i$ as $s$ goes from $0$ to $1$.
The function $h^K(x,s)$ is defined for all $x\in K$ with $\pi_{|K|}(x)\in [t-\epsilon, t+\epsilon]$.
Observe that 
$\pi_{|K|}(h^K(x,s))=(1-s)\pi_{|K|}(x)+st$,
so $h^K(\Phi(y),s)$ is defined for all $y\in A$ with $f(y)\in [0,\epsilon]$ and it is continuous for $y\in f^{-1}([0,\epsilon])\times [0,1]$.

Build the deformation retract $h$ from $f^{-1}([0, \epsilon])$ to $f^{-1}(0)$ as 
$h(x,s)=\Phi^{-1}(h^K(\Phi(x),s))$.

%Consider a cylindrical definable cell decomposition $\mathcal{C}_d$ of $\R^{d}$, compatible with $A$, but with a basis chosen so the first basis element is $v$ (so that the first coordinate is $g(x)$). Choose $\epsilon>0$ such that $(t, t+2\epsilon)$ and $(t-2\epsilon, t)$ are each completely contained in single cells of $\mathcal{C}_1$. Subdividing cells as needed, we can assume that $t, t+\epsilon$ and $t-\epsilon$ are all $0$-cells in $\mathcal{C}_1$. 

%Let $a_1<a_2<\ldots a_N$ denote the values that split $\mathcal{C}_1$ into different cells. By construction $a_j=t$ for some $j$. We will restrict our attention to cells $C\in\mathcal{C}_d$ whose projection $\pi_1(C)$ is either $(a_{j-1}, a_j)=(t-\epsilon, t)$, $\{a_j\}=\{t\}$ or $(a_j, a_{j+1})=(t-\epsilon, t)$. 

%We can build a deformation retraction from $\{x\in A\mid x\cdot v \in [t-\epsilon, t+\epsilon]\}$ to $LS(A,v,t)$ using this cell decomposition, retracting along the value of $f$. Here we are using $f^{-1}(s\epsilon)=g^{-1}(t-s\epsilon) \cup g^{-1}(t+s\epsilon)$.  
\end{proof}

\begin{prop}\label{prop:degree}
[Euler characteristic of slices from the Reeb graph]\label{prop:chilevel} Let $A\in \Surf(3)$ and let $v\in S^{2}$ and $t\in \R$ (see Figure~\ref{fig:injectivity3D}).% and choose $\varepsilon>0$ so that $t-\varepsilon$ and $t+\varepsilon$ are regular
%values of $h_v$ and $[t-\varepsilon,t+\varepsilon]$ contains no critical values for the Reeb graph other than possibly $t$.
Let $(\mathcal R_v(A),\tilde h_v)$ be the Reeb graph 
$A$ in direction $v$
(Definition~\ref{def:ReebGraph}), and let $V(\mathcal R_v(A))$ denote its vertex set.
Then
\begin{equation}\label{eq:chi_deg}
\chi\!\bigl(LS(A,v,t)\bigr)
\;=\;
\sum_{\substack{p\in V(\mathcal R_v(A))\\ \tilde h_v(p)=t}}\bigl(2-\deg(p)\bigr),
\end{equation}
where $\deg(p)$ is the local degree from Definition~\ref{def:ReebGraph}.

Equivalently, if $C(A,v,s)$ denotes the number of connected components of $LS(A,v,s)$, then
\begin{equation}\label{eq:chi}
\chi\!\bigl(LS(A,v,t)\bigr)
\;=\;
2\,C(A,v,t)\;-\;C(A,v,t-\varepsilon)\;-\;C(A,v,t+\varepsilon).
\end{equation}
for $\epsilon>0$ sufficiently small.
\end{prop}

\begin{proof}
Let $q_v:A\to \mathcal R_v(A)$ be the Reeb quotient map. By construction,
the connected components of $LS(A,v,t)$ are precisely the $\sim_v$--equivalence classes at height $t$,
hence are in bijection with the fiber $\tilde h_v^{-1}(t)\subset \mathcal R_v(A)$.
Write
\[
LS(A,v,t)=X_1\sqcup\cdots\sqcup X_r,
\qquad
p_j:=q_v(X_j)\in \tilde h_v^{-1}(t).
\]
Euler characteristic is additive over disjoint unions, so it suffices to compute $\chi(X_j)$ in terms
of the local behavior of $(\mathcal R_v(A),\tilde h_v)$ at $p_j$.

We have $h_v^{-1}[t-\epsilon, t+\epsilon]$ splits into disjoint connected components, $Y_1, \ldots, Y_r$ such that $X_j\subset Y_j$. For sufficiently small $\epsilon$ we have the restriction of $\mathcal R_v(A)$ to heights $[t-\epsilon, t+\epsilon]$ is the disjoint union of the Reeb graphs $\mathcal{R}_v(Y_j)$. 

From now on let us restrict to a single $Y_j$. By Lemma \ref{lem:retract}, for $\epsilon$ small enough we have $Y_j$ deformation retracts to $X_j$ and since both $X_j$ and $Y_j$ are compact this deformation retract implies $\chi(X_j)=\chi(Y_j)$ 
(note that the Euler characteristic used in the ECT and in its injectivity proofs all involve the definable version of Euler Characteristic which thankfully agrees with the standard version when the sets are compact). For $\epsilon$ small, $t-\epsilon$ and $t+\epsilon$ are not critical and we have $\partial Y_j=LS(Y_j, v, t-\epsilon) \cup LS(Y_j, v, t+\epsilon)$ is compact and homotopic to a disjoint union of $deg(p_j)$ copies of $S^1$. 

As the surface $Y_j$ is homotopic to a subset of $\R^2$ we further can infer that it has genus $0$. Thus, $Y_j$ is a surface with boundary with genus $0$ and $\deg(p_j)$ boundary curves, so from the classification of surfaces with boundary $\chi(Y_i)=2-\deg(p_i)$.

\end{proof}

As we can infer the Euler characteristic of every level set in every direction from the Reeb Transforms it follows from Theorem \ref{thm:ect-inj} we get the following corollary.

\begin{cor}\label{cor:S(3)_inj}
The Reeb Transform is injective when restricted to $\Surf(3)$.
\end{cor}

%\begin{proof}
%Finally, Theorem~\ref{thm:ect-inj} yields $A=B$.
%\end{proof}
The proof does not immediately generalise to surfaces embedded in $\R^d$ for $d>3$ because we have less control over the characterisation of $LS(A,v,t)$. When $d=3$ then $LS(A,v,t)$ is planar and each connected component has zero genus, and hence its Euler characteristic is determined by the number of boundary components. If we consider $d>3$ then $LS(A,v,t)$ could contain a connected component with non-zero genus. However, we conjecture it would be possible to modify the notion of collapsing direction to this setting. The goal would be to prove Proposition \ref{prop:degree} for a dense subset of directions. This could then be combined with Corollary \ref{cor:inj_dense} to show the Reeb Transform is injective when restricted to surfaces without boundary embedded in $\R^d$.

\begin{figure}
\center
\includegraphics[width=0.9\textwidth]{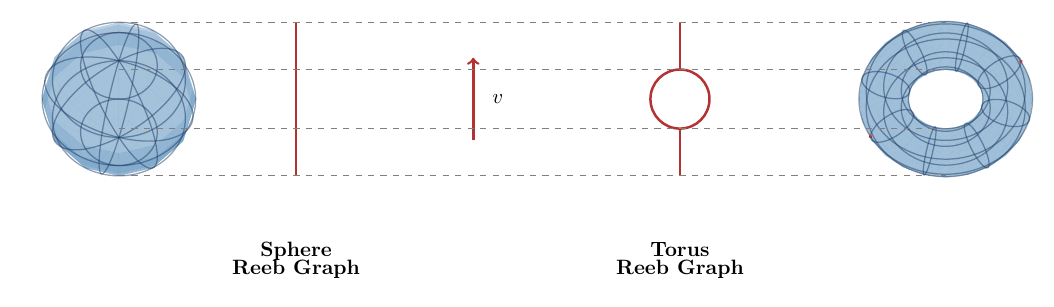}
\caption{Example of two different surfaces in $\mathbb{R}^3$ (a sphere and a torus) and a direction where their Reeb graphs differ. Such a direction must exist for any pair of distinct surfaces in $\R^3$}
\label{fig:injectivity3D}
\end{figure}

Another injectivity result relating to surfaces in $\R^3$ views the surface as a dividing hypersurface and considers the pair of the Reeb Transforms for the interior and exterior sets. 

\begin{lemma}\label{lem:surf_separates}
Let $A\in \Surf(3)$ be connected. Then $\R^{3}\setminus A$ has exactly two path-connected components, one bounded and one unbounded.
\end{lemma}

\begin{proof}
We use Alexander duality for the one-point compactification $S^{3}=\R^{3}\cup\{\infty\}$.
Throughout the proof, the notation $H_i(X;G)$ and $H^i(X;G)$ denotes singular homology and cohomology of $X$ with coefficients in an abelian group $G$. We also write $\widetilde H_i(X;G)$ for \emph{reduced} homology. 
In degree $0$, reduced homology removes the ubiquitous copy of $G$ coming from nonemptiness: if $X$ has $k$ connected components, then
\[
H_0(X;G)\cong G^{\,k}
\qquad\text{and}\qquad
\widetilde H_0(X;G)\cong G^{\,k-1}.
\]
In particular, $\widetilde H_0(X;G)\cong G$ holds if and only if $X$ has exactly two connected components.

Since $A$ is definable and compact, it is locally contractible, so Alexander duality applies.
Using coefficients $G=\Z/2$ (to avoid orientability issues), Alexander duality gives
\[
\widetilde H_{0}(S^{3}\setminus A;\Z/2)\ \cong\ H^{2}(A;\Z/2).
\]
Because $A$ is homeomorphic to a connected compact surface, its top cohomology with $\Z/2$ coefficients satisfies $H^{2}(A;\Z/2)\cong \Z/2$. Therefore
\[
\widetilde H_{0}(S^{3}\setminus A;\Z/2)\cong \Z/2,
\]
which implies that $S^{3}\setminus A$ has exactly two connected components.

Finally, $\R^{3}\setminus A$ is definable and hence locally path-connected, so its connected components coincide with its path-connected components. One component contains $\infty$ and corresponds to the unbounded component of $\R^{3}\setminus A$. The other component is bounded since $A$ is compact.
\end{proof}

We use the two connected components of the complement of a connected surface $A\in\Surf(3)$
(Lemma~\ref{lem:surf_separates}) to define the interior and exterior of $A$.

\begin{definition}[Interior and exterior regions]\label{def:int_ext}
Let $A\in \Surf(3)$ be connected. By Lemma~5, the complement $\R^{3}\setminus A$ has exactly two
path-connected components: one bounded and one unbounded. Denote the bounded component by
$\omega_{\mathrm{int}}$ and the unbounded component by $\omega_{\mathrm{ext}}$.

We call $\omega_{\mathrm{int}}$ the \emph{open interior region} of $A$ and $\omega_{\mathrm{ext}}$ the
\emph{open exterior region} of $A$. We define the \emph{closed interior region} and
\emph{closed exterior region} of $A$ by
\[
A^{\mathrm{int}} := \overline{\omega_{\mathrm{int}}}
\qquad\text{and}\qquad
A^{\mathrm{ext}} := \overline{\omega_{\mathrm{ext}}},
\]
where the closure is taken in $\R^{3}$.
\end{definition}

%We have  $A^{\mathrm{int}}$ is the closed, bounded $3$‑dimensional subset whose boundary is $A$, while $A^{\mathrm{ext}}$ is the closure of the unbounded complement together with the same boundary $A$.
 
\begin{figure}
\center
\includegraphics[width=0.8\textwidth]{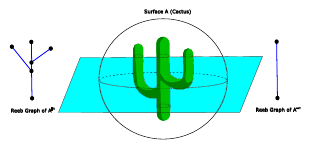}
\caption{Illustrating Theorem~\ref{theorem:ext_int}. 
A dividing surface $A\subset\mathbb{R}^3$ splits space into an interior ($A^{int}$) and exterior ($A^{ext}$). 
By computing the Reeb Transforms of these two extended domains, $\mathcal{RT}(A^{int})$ and $\mathcal{RT}(A^{ext})$, 
one obtains all the necessary information to reconstruct $A$ itself.}
\label{fig:interiorRegions}
\end{figure}

\begin{thm}\label{theorem:ext_int}
Let $A\in \Surf(3)$ be connected, and let $A^{\mathrm{int}}$ and $A^{\mathrm{ext}}$ be the closed
interior and exterior regions from Definition~\ref{def:int_ext}.%
\footnote{Although $A^{\mathrm{ext}}$ is unbounded, it is closed and definable (being the closure of
a definable open set). In this section we use the directional Reeb quotient/graph (and hence the
Reeb Transform) for such closed definable sets whenever the quotient is a definable graph. This holds for $A^{\mathrm{ext}}$ in our setting because $A^{\mathrm{ext}}$ is closed and definable and,
for each $v\in S^{2}$ and $t\in\R$, the slice $A^{\mathrm{ext}}\cap\{x\cdot v=t\}$ is a definable subset
of the plane and hence has finitely many connected components. In particular, the Reeb quotient
$\mathcal R_v(A^{\mathrm{ext}})$ is a definable one-dimensional graph-like space (possibly with unbounded
edges); see~\cite{basu2022reeb} and references therein.}
Then $A$ (as a subset of $\R^3$) is uniquely determined by the pair
$\bigl(\mathcal{RT}(A^{\mathrm{int}}),\,\mathcal{RT}(A^{\mathrm{ext}})\bigr)$(see Figure~\ref{fig:interiorRegions}).
\end{thm}

\begin{proof}
Let $B\in \Surf(3)$ be another connected surface such that
\[
\mathcal{RT}(A^{\mathrm{int}})\cong\mathcal{RT}(B^{\mathrm{int}})
\qquad\text{and}\qquad
\mathcal{RT}(A^{\mathrm{ext}})\cong\mathcal{RT}(B^{\mathrm{ext}}).
\]
We will prove $A=B$.

For any set $X\subset \R^3$, direction $w\in S^{2}$, and height $s\in\R$, write
\[
C(X,w,s)\ :=\ \text{the number of connected components of }LS(X,w,s).
\]

\smallskip
\noindent\textbf{Step 1: A component identity valid for all heights.}
Fix $w\in S^{2}$ and $s\in\R$, and let
\[
P:=\{x\in\R^{3}:x\cdot w=s\}\cong \R^{2}.
\]
Set
\[
X_{\mathrm{int}}:=LS(A^{\mathrm{int}},w,s)=A^{\mathrm{int}}\cap P,\qquad
X_{\mathrm{ext}}:=LS(A^{\mathrm{ext}},w,s)=A^{\mathrm{ext}}\cap P,\qquad
K:=LS(A,w,s)=A\cap P.
\]
Since $A^{\mathrm{int}}\cup A^{\mathrm{ext}}=\R^{3}$ and $A^{\mathrm{int}}\cap A^{\mathrm{ext}}=A$,
intersecting with $P$ yields
\[
P = X_{\mathrm{int}}\cup X_{\mathrm{ext}}
\qquad\text{and}\qquad
K = X_{\mathrm{int}}\cap X_{\mathrm{ext}}.
\]

% If $K=\emptyset$, then $P\subset \R^{3}\setminus A$, so
% $P\cap \omega_{\mathrm{int}}$ and $P\cap \omega_{\mathrm{ext}}$ are disjoint open sets covering $P$.
% Since $P$ is connected, exactly one of them is empty. It follows that either
% $X_{\mathrm{int}}=\emptyset$, and hence
% \[
% C(A,w,s)=0 = C(A^{\mathrm{int}},w,s)+C(A^{\mathrm{ext}},w,s)-1.
% \]
% Thus~\eqref{eq:component_identity} holds in this case. Hence assume $K\neq\emptyset$.

If $K=\emptyset$, then $P\subset \R^{3}\setminus A$. Since $P$ is connected, it must lie entirely in one of the two connected components $\omega_{\mathrm{int}}$ or $\omega_{\mathrm{ext}}$ of $\R^{3}\setminus A$.
However, $\omega_{\mathrm{int}}$ is bounded, whereas $P\cong\R^{2}$ is unbounded, so $P$ cannot be contained in $\omega_{\mathrm{int}}$. 
Hence
\[
P\subset \omega_{\mathrm{ext}}.
\]
Since $P\cap A=\emptyset$, it follows that
\[
X_{\mathrm{int}}=\emptyset
\qquad\text{and}\qquad
X_{\mathrm{ext}}=P.
\]
Therefore,
\[
C(A,w,s)=0,
\qquad
C(A^{\mathrm{int}},w,s)=0,
\qquad
C(A^{\mathrm{ext}},w,s)=1,
\]
and hence
\[
C(A,w,s)
=
C(A^{\mathrm{int}},w,s)
+
C(A^{\mathrm{ext}},w,s)
-1.
\]
Thus~\eqref{eq:component_identity} holds when $K=\emptyset$.
Henceforth assume $K\neq\emptyset$.

Because $X_{\mathrm{int}},X_{\mathrm{ext}},K$ are closed definable subsets of the plane $P$,
they are locally compact and locally contractible by \cite[Theorem ~4.4]{coste1999introduction}. There is a definable homeomorphism $\Phi$ from $\R^3$ to a bounded subset. The sets $\Phi(X_{\mathrm{int}})$, $\Phi(X_{\mathrm{ext}})$ and $\Phi(K)$ are each homeomorphic to finite simplicial complexes and hence each admits an open $\varepsilon$-neighborhood ($N_\varepsilon(\Phi(X_{\mathrm{ext}}))$, $N_\varepsilon(\Phi(X_{\mathrm{ext}}))$ and $N_\varepsilon(\Phi(K))$ respectively) that strong deformation retracts onto it \cite[Proposition ~A.5]{Hatcher2002}.

Setting $U_{\mathrm{int}} = \Phi^{-1}(N_\varepsilon(\Phi(X_{\mathrm{int}}))$ and $U_{\mathrm{ext}} = \Phi^{-1}(N_\varepsilon(\Phi(X_{\mathrm{ext}})))$ we then, by conjugating the previous deformation retracts with $\Phi$, we also have each of
$X_{\mathrm{int}} \hookrightarrow U_{\mathrm{int}}$ and
$X_{\mathrm{ext}} \hookrightarrow U_{\mathrm{ext}}$ is a strong deformation
retract. Moreover $U_{\mathrm{int}} \cup U_{\mathrm{ext}} \supseteq
X_{\mathrm{int}} \cup X_{\mathrm{ext}} = \R^2 $.
By construction
\begin{align*}
U_{\mathrm{int}} \cap U_{\mathrm{ext}} &=\Phi^{-1}(N_\varepsilon(\Phi(X_{\mathrm{int}})) \cap \Phi^{-1}(N_\varepsilon(\Phi(X_{\mathrm{ext}})))\\
&=\Phi^{-1}(N_\varepsilon(\Phi(X_{\mathrm{int}})\cap (N_\varepsilon(\Phi(X_{\mathrm{ext}}))\\
&=\Phi^{-1}(N_\varepsilon(X_{\mathrm{int}})\cap X_{\mathrm{ext}})\\
&=\Phi^{-1}(N_\varepsilon(K))
\end{align*}
and hence $U_{\mathrm{int}} \cap U_{\mathrm{ext}}$ strong deformation retracts onto $K$.

Apply the Mayer--Vietoris exact sequence in singular homology to the open cover
$P=U_{\mathrm{int}}\cup U_{\mathrm{ext}}$ (see, e.g., \cite[Sec.~2.2]{Hatcher2002}).
Since $P\cong \R^{2}$ is contractible, $H_{1}(P)=0$ and $H_{0}(P)\cong \Z$, so exactness yields a
short exact sequence
\[
0\longrightarrow H_{0}(U_{\mathrm{int}}\cap U_{\mathrm{ext}})
\longrightarrow H_{0}(U_{\mathrm{int}})\oplus H_{0}(U_{\mathrm{ext}})
\longrightarrow H_{0}(P)\longrightarrow 0.
\]
Taking ranks and using $\mathrm{rank}\,H_{0}(Y)=C(Y)$ for any space $Y$ with finitely many connected
components (including $Y=\emptyset$), we obtain
\[
C(U_{\mathrm{int}}\cap U_{\mathrm{ext}})=C(U_{\mathrm{int}})+C(U_{\mathrm{ext}})-1.
\]
By the homotopy equivalences above, we have
\[
C(U_{\mathrm{int}}\cap U_{\mathrm{ext}})=C(K),\qquad
C(U_{\mathrm{int}})=C(X_{\mathrm{int}}),\qquad
C(U_{\mathrm{ext}})=C(X_{\mathrm{ext}}),
\]
and hence
\[
C(K)=C(X_{\mathrm{int}})+C(X_{\mathrm{ext}})-1,
\]
i.e.,
\begin{equation}\label{eq:component_identity}
C(A,w,s)\ =\ C(A^{\mathrm{int}},w,s)\ +\ C(A^{\mathrm{ext}},w,s)\ -\ 1
\end{equation}
for all $w\in S^{2}$ and all $s\in\R$.

\smallskip
\noindent\textbf{Step 2: Recovering $\chi(LS(A,w,t))$ from the pair of transforms.}
Now fix a direction $w$. 
%such that the height functions $h_w|_{A}$ and $h_w|_{B}$ are Morse with
%distinct critical values. Fix any $t\in\R$, and choose $\varepsilon>0$ so that $t-\varepsilon$%
%and $t+\varepsilon$ are regular values of both $h_w|_{A}$ and $h_w|_{B}$ and such that there are no
%other critical values of either function in $(t-\varepsilon,t+\varepsilon)$.
Then Proposition~\ref{prop:chilevel} gives
\[
\chi\!\bigl(LS(A,w,t)\bigr)\;=\;
2\,C(A,w,t)\;-\;C(A,w,t-\varepsilon)\;-\;C(A,w,t+\varepsilon),
\]
and likewise for $B$.

For each $s\in\{t,\,t-\varepsilon,\,t+\varepsilon\}$, the values
$C(A^{\mathrm{int}},w,s)$ and $C(A^{\mathrm{ext}},w,s)$ are determined by the Reeb graphs appearing
in $\mathcal{RT}(A^{\mathrm{int}})$ and $\mathcal{RT}(A^{\mathrm{ext}})$ (as the cardinalities of the
fibers at height $s$), and similarly for $B$. Hence~\eqref{eq:component_identity} determines
$C(A,w,s)$ and $C(B,w,s)$ for these heights. Since the assumed isomorphisms of Reeb transforms imply
\[
C(A^{\mathrm{int}},w,s)=C(B^{\mathrm{int}},w,s)
\quad\text{and}\quad
C(A^{\mathrm{ext}},w,s)=C(B^{\mathrm{ext}},w,s)
\qquad\text{for all such }s,
\]
we obtain $C(A,w,s)=C(B,w,s)$ for $s\in\{t,\,t-\varepsilon,\,t+\varepsilon\}$, and consequently
\begin{equation}\label{eq:slice-eq-on-dense}
\chi\!\bigl(LS(A,w,t)\bigr)=\chi\!\bigl(LS(B,w,t)\bigr)
\qquad\text{for all }t\in\R
\end{equation}
for every direction $w$.% for which both $h_w|_A$ and $h_w|_B$ are Morse with distinct critical values.

Finally, Theorem~\ref{thm:ect-inj} yields $A=B$.
\end{proof}

\section{Conclusion and future work}

In this work, we have initiated the study of the Reeb transform.  While not a complete invariant of an input shape, the utility of Reeb graphs in practice motivates their study, and our results show that they are discriminative and potentially useful for a wide class of low dimensional settings.

Our work suggests several interesting possibilities for future study.  First, there are well established  connections between the more widely used mapper graphs and Reeb graphs, in particular that mapper graphs in some sense converge in the limit to Reeb graphs~\cite{Munch2016,Brown2020}.  That said, mapper graphs are considerably more unstable, with many parameter selection choices affecting the resulting graph~\cite{Carriere2018}.  We conjecture that mapper transforms would share similar properties as we studied here, but different tools would be necessary given that mapper is typically build over point cloud data.  Given the utility of mapper graphs in general but also the necessarily very different tools that would be used, we leave this interesting question for future work.

From a more applied perspective, the PHT and ECT have both proven useful in applications, so it would be interesting to study experimental evaluations of Reeb graphs for shape analysis pipelines.  Previous work has studied the complexity of maintaining dynamic Reeb graphs, where the input shape changes in some controlled way and the algorithm updates the graph at each step without recomputing~\cite{Agarwal2018}.  It seems likely that similar approaches could be adapted to the setting where the height function varies continuously, yielding fast algorithms and data structures.  Depending on the data set, this might also connect to the mapper graph algorithms above, as mapper graphs are more commonly used for point cloud data, where no underlying simplicial complex is known.

Beyond Reeb graphs, recent work has also studied Reeb spaces, where multiple functions on the input space project it into a lower dimensional structure; however, to the best of our knowledge there is no known work on Reeb space transforms.  Aside from  using height functions, one could also consider intrinsic functions like curvature or affine Grassmannian transforms (i.e. distance to a point, line, or hyperplane)~\cite{onus2024shovingtubesshapesgives}, and study the discriminativity of the resulting Reeb space or annotated Reeb graph.

\section*{Conflict of Interest Statement}
On behalf of all authors, the corresponding author states that there is no conflict of interest.
%\backmatter

%\bmhead{Acknowledgements}
%Add grants
\bibliographystyle{plain}
\bibliography{ReebTransform}

\end{document}